\documentclass{article}
\usepackage[utf8]{inputenc}
\usepackage[T2A]{fontenc}
\usepackage[english]{babel}
\inputencoding{utf8}

\usepackage{amsmath,amssymb,amsthm}

\usepackage[pdfpagemode=UseNone,bookmarks=false,pdfpagelayout=SinglePage,pdfstartview={XYZ null null 1},
pdfhighlight=/P,colorlinks=true,linkcolor=blue,citecolor=blue,urlcolor=blue]{hyperref}

\advance\textwidth 20mm  \advance\oddsidemargin -10mm
\advance\textheight 20mm \advance\topmargin -15mm

\theoremstyle{plain}
\newtheorem{theorem}{Theorem}
\newtheorem{lemma}[theorem]{Lemma}
\theoremstyle{definition}
\newtheorem{remark}{Remark}

\DeclareMathOperator{\diag}{diag}
\DeclareMathOperator{\Mat}{Mat}

\title{On matrix Painlev\'e II equations}

\date{6 December 2020}

\author{V.E.~Adler\thanks{L.D.~Landau Institute for Theoretical Physics, Chernogolovka, Russian Federation.}, 
\,\,V.V.~Sokolov$^*$\thanks{Federal University of ABC, Santo Andr\'e, Sao Paulo, Brazil. E-mail: vsokolov@landau.ac.ru}}

\begin{document}
\maketitle

\begin{abstract}
The Painlev\'e--Kovalevskaya test is applied to find three matrix versions of the Painlev\'e II equation. All these equations are interpreted as group-invariant  reductions of integrable matrix evolution equations, which makes it possible to construct isomonodromic Lax pairs for them.
\medskip

\end{abstract}

%-------------------------------------------------------------------------------
\section{Introduction}

The second of the six famous Painlev\'e equations reads
\begin{equation}\label{P2}
 y''= 2 y^3 +z y + a,
\end{equation}
where $'$ denotes the derivative with respect to $z$ and $a$ is an arbitrary complex constant. A matrix version of the Painlev\'e--Kovalevskaya test was proposed in \cite{Balandin_Sokolov_1998}, where it was proved that it holds for \eqref{P2} in the case when $y(z)$ is a matrix of arbitrary size $n\times n$ and $a$ is a scalar matrix.   

One of possible generalizations of the matrix equation \eqref{P2} is as follows. It is clear that in the non-commutative case one can change the principal differential-homogeneous part $y''= 2 y^3$ of this equation by adding the term of the same weight $\kappa [y, y']$, where $\kappa\in{\mathbb C}$. In particular, from the results of Section 2 it follows that the equation
$$
 y''= \kappa [y,y'] +2y^3 +z y +a,\qquad a\in \mathbb C
$$
satisfies the matrix Painlev\'e--Kovalevskaya test if and only if $\kappa = 0, \pm 1, \pm 2.$ Note that equations with opposite signs of $\kappa$ are related by the changes $y\to -y$, $a\to -a$ or $y\to y^T$.

Another direction for generalizations was suggested by the results of \cite{Retakh_Rubtsov_2010}, where the equation P$_2$ with matrix coefficients appeared. More precisely, in this paper the term linear in $y$ was written as $\tfrac{1}{2}(zy+yz)$ where $z$ denoted a non-commutative dependent variable such that $z'= 1$. For our purposes it is more convenient to replace this variable with $z+2b$, where $z$ is a commutative independent variable and $b$ is a matrix constant.

In this paper, we study matrix generalizations of the P$_2$ equation of the general form
\begin{equation}\label{gP2}
 y''= \kappa[y,y'] +2y^3 +z y +b_1 y +y b_2 +a,
\end{equation}
where $a$, $b_1$ and $b_2$ are matrix constants and $\kappa$ is a scalar constant.

\begin{remark} 
Equation (\ref{gP2}) is invariant under the change
\begin{equation}\label{bbb}
b_1\to b_1+\beta_1I,\qquad  b_2\to b_2+\beta_2I,\qquad z\to z-\beta_1-\beta_2, 
\end{equation} 
where $\beta_i\in{\mathbb C}$ and $I$ is the identity matrix.
\end{remark} 

In the second section, by use of the matrix version of Painlev\'e--Kovalevskaya test, we examine the sets of $\kappa$, $a$, $b_1$ and $b_2$ which are candidates for integrable cases and prove the following statement.

\begin{theorem}
The equation \eqref{gP2} satisfies the Painlev\'e--Kovalevskaya test if and only if it is reduced by the transformation \eqref{bbb} to one of the following cases:
\begin{alignat}{2}
\label{P20}\tag*{P$_2^0$}
 & y''= 2y^3+zy+by+yb+\alpha I, &\qquad& \alpha\in{\mathbb C},~~ b\in\Mat_n,\\
\label{P21}\tag*{P$_2^1$}
 & y''= \pm[y,y']+2y^3+zy+a, && a\in\Mat_n,\\
\label{P22}\tag*{P$_2^2$}
 & y''= \pm2[y,y']+2y^3+zy+by+yb+a, && a,b\in\Mat_n,~~ [b,a]=\pm2b.
\end{alignat}
\end{theorem}

Note that \ref*{P20} is exactly the equation from \cite{Retakh_Rubtsov_2010}. The equations \ref*{P21} and \ref*{P22} appear to be new.

In Section 3, the integrability of the found cases is substantiated by construction of the isomonodromic Lax pairs
\begin{equation}\label{AB}
 A'=B_\zeta+[B,A],
\end{equation}
where $\zeta$ is a spectral parameter, $A(z,\zeta)$ and $B(z,\zeta)$ are $2\times2$ matrices with non-commutative entries. For equation \ref*{P20}, such a representation was found in \cite{Irfan_2012}.

One of the methods for obtaining representations \eqref{AB} is based on the observation that group-invariant solutions of evolution equations admitting zero curvature representations satisfy ordinary differential equations of the Painlev\'e type. In this case, the isomonodromic representation is obtained from the zero curvature representation by a standard procedure. In particular, it is well known that the P$_2$ equation corresponds to Galilean-invariant solutions of the nonlinear Schr\"odinger equation, as well as to scale-invariant solutions of the mKdV equation. Here we generalize these reductions to the matrix case.

One source of non-commutative coefficients in \eqref{gP2} is the arbitrary matrices contained in the symmetry groups of non-Abelian \cite{Marchenko_1986, OS_1998a} evolutionary equations. For example, the matrix nonlinear Schr\"odinger equation
\begin{equation}\label{NLS}
     u_t=u_{xx}+2uvu ,\qquad v_t= -v_{xx}-2vuv   
\end{equation}
admits transformations of the form $u\to b_1 u b_2$, $v\to b_2^{-1} v b_1^{-1}$, where $b_i\in\Mat_n.$ Two well-known matrix generalizations of the mKdV equation \cite{Marchenko_1986, Khalilov_Khruslov_1990}
\begin{gather*}
 u_t= u_{xxx}-3u^2u_x-3u_xu^2,\\[1.5mm]
 u_t= u_{xxx}+3[u,u_{xx}]-6uu_xu
\end{gather*}
obviously admit the transformation group $u \to b u b^{-1}.$

Another source of matrix coefficients in \eqref{gP2} is related with matrix coefficients in integrable evolution equations themselves. In this paper, we use the following integrable version of the matrix mKdV equation:
$$
 u_t= u_{xxx}+3[u, u_{xx}]-6u u_x u + (u_x+u^2) b + b (u_x-u^2), \qquad b\in\Mat_n.
$$
Its zero curvature representation is given in Section 3. We do not know whether this generalization appeared in the literature.

In Section 3, we show that equations P$_2^i$ from Theorem 1 are obtained from the above matrix evolution equations by some group-invariant reductions. This made possible to find the isomonodromic Lax pairs \eqref{AB} for the matrix Painlev\'e II equations.

\section{Proof of Theorem 1}\label{s:P2-test}
 
\subsection{The Painlev\'e--Kovalevskaya matrix test}

The Painlev\'e--Kovalevskaya matrix test \cite{Balandin_Sokolov_1998} for equation (\ref{gP2}) is based on counting of arbitrary scalar constants in a formal solution of the form
\begin{equation}\label{y}
 y=\frac{p}{z-z_0}+c_0+c_1(z-z_0)+\cdots\,, \qquad 
 p,c_j \in\Mat_n,\quad z_0\in{\mathbb C}.
\end{equation}
Here $z_0$ is one of these arbitrary constants. In order for the series $y$ to represent a generic solution, it is necessary that the matrices $p$ and $c_j$ contain additionally $2n^2-1$ arbitrary constants. We assume that the Painlev\'e--Kovalevskaya test is fulfilled if such matrices exist. Note that there may also exist other formal solutions of the form \eqref{y} that contain fewer arbitrary constants.

\begin{remark} 
For any nondegenerate matrix $T$, the series $T y T^{-1}$ satisfies equation (\ref{gP2}), where $b_i\to \bar b_i \stackrel{def}{=} T b_i T^{-1}$ and $a \to \bar a\stackrel{def}{=}T a T^{-1}$. Hence, equation (\ref{gP2}) satisfies the Painlev\'e--Kovalevskaya test simultaneously with the equation corresponding to the coefficients $\bar b_i$ and $\bar a$.
\end{remark} 

Substituting the series into the equation and collecting the coefficients at powers of $z-z_0$, we obtain relations of the form
\begin{gather}
\label{ep3}
 p^3=p,\\ 
\label{eLc}
 L_{\frac{j+1}{2}\kappa}(c_j) -\frac{j(j-1)}{2}c_j =f_j(z_0,p,c_0,\dots,c_{j-1}),\quad j\ge 0, 
\end{gather}
where 
\[
 L_\sigma(c) \stackrel{def}{=} p^2c + pcp+ cp^2 +\sigma(pc-cp).
\]
The first few functions $f_i$ are easy to compute explicitly, for instance,
\[
 f_0=0,\qquad f_1=-pc^2_0-c_0pc_0-c^2_0p-\frac{1}{2}(z_0p+b_1p+pb_2).
\]
Since $f_j$ in the right-hand sides of (\ref{eLc}) do not contain $c_j$, the matrices $c_j$ can be calculated from these equations recursively. If the linear operator in the left-hand side is invertible, then $c_j$ is uniquely determined. If, for some $j$, the corresponding operator is degenerate then, firstly, the answer contains arbitrary constants in the amount equal to the dimension of the kernel; secondly, the solvability conditions lead to some constraints on $f_j$, from which the conditions on the parameters $b_1$, $b_2$ and $a$ should be extracted.

We start from calculating the possible total number of arbitrary constants. The number of arbitrary constants in the matrix $p$ is equal to the dimension of the orbit of its Jordan form. It is easy to see that the Jordan form of any matrix $p$ satisfying \eqref{ep3} is
\[
 p=\diag(E_k,\, -E_m,\, 0_{n-k-m}).
\]  
When a group acts on a manifold, the dimension of the orbit of a point is equal to the difference of the manifold dimension and the dimension of the stabilizer of this point, that is, the subgroup that leaves the point fixed. In our case, the dimension of the manifold is $n^2$, the stabilizer consists of non-degenerate matrices which commute with $p$. It is easy to see that its dimension is $k^2+m^2+(n-k-m)^2$; whence it follows that the dimension of the orbit of $p$ is equal to $2m(n-m)+2k(n-k)-2km$.

Next, we calculate the dimensions of the eigenspaces of the linear operator $L_\sigma:\Mat_n\to\Mat_n$. 

\begin{lemma}\label{l:dim} 
The eigenvalues of the operator $L_\sigma$ belong to the set
\[
 \lambda_0=0,\quad \lambda_{\pm 1}=1\pm \sigma, \quad \lambda_{\pm 2}=1\pm 2 \sigma, \quad \lambda_3=3.
\]  
The space $\Mat_n$ decomposes into the direct sum of eigenspaces
\[
 \Mat_n=V_0\oplus V_{-1}\oplus V_1\oplus V_{-2}\oplus V_2\oplus V_3,\qquad L_\sigma V_i= \lambda_iV_i,
\]
with dimensions
\begin{gather*}
\dim V_0 = (n-k-m)^2, \qquad  \dim V_{\pm1} = k(n-k)+m(n-m) - 2 k m,\\  
\dim V_{\pm2} = k m, \qquad  \dim V_3 = k^2+m^2.
\end{gather*}
In the case when the eigenvalues coincide, the dimensions of the corresponding eigenspaces add up.
\end{lemma}

\begin{proof} 
We represent $c$ as a $3\times3$ block matrix with the block sizes determined by the Jordan form of $p$:
\[
 c=\bordermatrix{  
    & k      & m      & l \cr
  k & c_{11} & c_{12} & c_{13} \cr
  m & c_{21} & c_{22} & c_{23} \cr
  l & c_{31} & c_{32} & c_{33}},
 \qquad l=n-k-m.
\]
Then it is easy to check that
\[
 L_\sigma(c)=\begin{pmatrix}
  3c_{11} & (1+2\sigma)c_{12} & (1+\sigma)c_{13}\\ 
  (1-2\sigma)c_{21} & 3c_{22} & (1-\sigma)c_{23}\\ 
  (1-\sigma)c_{31}  & (1+\sigma)c_{32} & 0 
 \end{pmatrix},
\]
that is, each eigenspace consists of one or two blocks of the matrix $c$. The dimensions of the eigenspaces are calculated by summing the block sizes.
\end{proof}

The operator in the left hand side of (\ref{eLc}) degenerates if one of the eigenvalues of the operator $L_{\frac{j+1}{2}\kappa}$ coincides with $\frac{j(j-1)}{2}$ (we call this the resonance condition). In particular, the eigenvalue $\lambda_0=0$ appears twice: for $j=0$ and for $j=1$, and the eigenvalue $\lambda_3=3$ appears once for $j=2$. The conditions that the eigenvalues $\lambda_{\pm1}$ and $\lambda_{\pm2}$ are resonant for some index $j$, are given, respectively, by equations
\[
 1\pm \frac{\kappa (j+1)}{2}=\frac{j(j-1)}{2} \quad\text{and}\quad
 1\pm \kappa (j+1)=\frac{j(j-1)}{2}.
\]
Cancelling this by $j+1$ (recall that $j\ge0$), we obtain the resonance conditions
\[
 \lambda_{\pm1}:~ j=2 \pm \kappa,\qquad
 \lambda_{\pm2}:~ j=2 \pm 2\kappa.
\]
We see that for a fixed $\kappa$ any of eigenspaces $V_{\pm 1}$, $V_{\pm 2}$ can occur at its resonant $j$ no more than once. This makes possible to estimate from above their total dimension, that is, the number of possible arbitrary constants. Namely, if both $\lambda_{\pm1}$ and $\lambda_{\pm2}$ appear in (\ref{eLc}), then the total sum of dimensions is
\begin{multline*}
 \qquad 2\dim V_0+\dim V_{-1}+\dim V_1+\dim V_{-2}+\dim V_2+\dim V_3\\
 = \dim V_0 +n^2 = (n-k-m)^2+n^2.\qquad
\end{multline*}
By adding the dimension of the orbit of $p$ we obtain that the total number of arbitrary constants is equal to $2n^2-k^2-m^2$. This is equal to $2n^2-1$ only for $k=1$, $m=0$ or for $k=0$, $m=1$. The second case is reduced to the first one by the change $y\to -y$. In the first case the Jordan form of $p$ is $\diag(1,0,\dots,0)$ and its orbit ${\cal O}$ consists of the matrices of the form $u \, v^T$, where $u$ and $v$ are column vector such that $u^T v=1.$ 

Thus, we assume that $p\in {\cal O}$. Then the eigenvalues $\lambda_{\pm 2}$ from the Lemma \ref{l:dim} disappear. To provide the required number of arbitrary constants, the series \eqref{y} must exist for the general element of the orbit. More precisely, we require that:
\begin{itemize}
    \item for any column vectors $u$ and $v$ such that $u^T v=1$, there exists a formal solution of the form
    $$
 y=\frac{u\,v^T}{z-z_0}+c_0+c_1(z-z_0)+\cdots\,, \qquad 
 p,c_j \in\Mat_n,\quad z_0\in{\mathbb C},
    $$
such that its coefficients $c_i$ contain in total $2n^2-2n+1$ arbitrary constants.
\end{itemize}

{\bf Proof of Theorem 1.} The above counting of arbitrary constants shows that a necessary condition for passing the Painlev\'e--Kovalevskaya test is that {\em both} eigenvalues $\lambda_{\pm 1}$ must be resonant. The resonance conditions for $j=2\pm\kappa$ imply that the parameter $\kappa$ must be such that both numbers $2\pm\kappa$ are non-negative integers. This leaves the admissible values
\[
 \kappa=0,\,\pm 1,\,\pm 2.
\] 
The case $\kappa<0$ is reduced to $\kappa>0$ by the change $y\to y^T$. Therefore, further it is sufficient to analyze equations (\ref{eLc}) for each of the three cases $\kappa=0,1$ and 2. 

\subsection{Case $\kappa = 0$}\label{kap0}

Let $p=\diag(1,0,\dots,0)$. In this case the operator $L_\sigma$ is of the form 
\[
 L_\sigma:\quad 
 \bordermatrix{ & 1 & n-1 \cr
  ~~~1 & c_{11} & c_{12}\cr 
  n-1  & c_{21} & c_{22}} 
 \mapsto
 \begin{pmatrix}
  3c_{11} &  (1+\sigma)c_{12}\\ 
  (1-\sigma)c_{21} & 0 
 \end{pmatrix}.
\]
In order to analyze equations (\ref{eLc}), we write them block-wise by dividing the involved matrices into blocks of appropriate sizes:
\[
 c_j=\begin{pmatrix}
   c_{j,11} & c_{j,12} \\
   c_{j,21} & c_{j,22}
 \end{pmatrix},\quad
 b_1=\begin{pmatrix}
   b_{1,11} & b_{1,12} \\
   b_{1,21} & b_{1,22}
 \end{pmatrix},\quad
 b_2=\begin{pmatrix}
   b_{2,11} & b_{2,12} \\
   b_{2,21} & b_{2,22}
 \end{pmatrix},\quad
 a=\begin{pmatrix}
   a_{11} & a_{12} \\
   a_{21} & a_{22}
 \end{pmatrix}.
\]
Conditions on the matrices $b_1$, $b_2$ and $a$ arise from the requirement of existence of solutions for inhomogeneous linear systems \eqref{eLc} for the resonance values of $j$. In this case, the homogeneous system has a nontrivial solution and the conditions of the Kronecker--Capelli theorem must be satisfied.

Let us consider equations (\ref{eLc}) for $j=0,1,2$ and $3$. For $j=0$ we find $c_{0,11}=c_{0,12}=c_{0,21}=0$ and the block $c_{0,22}$ is arbitrary. 

Next, for $j=1$ we obtain
\[
 c_{1,11}=-\frac{1}{6}(z_0+b_{1,11}+b_{2,11}),\qquad
 c_{1,12}=-\frac{1}{2}b_{2,12},\qquad
 c_{1,21}=-\frac{1}{2}b_{1,21}
\]
and $c_{1,22}$ is arbitrary. So far, no conditions for the coefficients $b_1,b_2$ and $a$ appear. 

For $j=2$ we find, by substituting the obtained values, that
\[
 c_{2,11}=-\frac{1}{4}(1+a_{11}),\qquad
 c_{2,22}=\frac{1}{2}(a_{22}+b_{1,22}c_{0,22}+c_{0,22}b_{2,22}+z_0c_{0,22})+c^3_{0,22},
\]
and the blocks $c_{2,12}$ and $c_{2,21}$ are arbitrary. In addition, we obtain the restrictions on the coefficients in the form of the relations
$$
 (b_{2,12}-b_{1,12})c_{0,22}=a_{12},\qquad c_{0,22}(b_{1,21}-b_{2,21})=a_{21}.
$$
Since the block $c_{0,22}$ is arbitrary, this implies that
\begin{equation}\label{cond1}
 b_{1,12}=b_{2,12},\qquad b_{1,21}=b_{2,21},\qquad a_{12}=a_{21}=0.
\end{equation}

Finally, for $j=3$ it is sufficient to write down the condition for the $1\times1$ block
\begin{equation}\label{cond2}
 b_{1,12}b_{1,21}-2b_{2,12}b_{1,21}+b_{2,12}b_{2,21}=0,
\end{equation}
because all other equations are solved uniquely with respect to $c_{3,12}$, $c_{3,21}$ and $c_{3,22}$. The condition \eqref{cond2} follows from \eqref{cond1}. For $j>3$, there are no resonances and therefore all equations for $c_j$ are solved uniquely.

Thus, we proved that equations (\ref{eLc}) with $\kappa=0$ are solved with arbitrary blocks $c_{0,22}$, $c_{1,22}$, $c_{2,12}$, $c_{2,21}$ and $c_{3,11}$, which gives, as one can easily see, $2(n-1)^2+2(n-1)+1$ arbitrary constants. Moreover, we obtained the solvability conditions in the form of relations (\ref{cond1}) which mean that the matrices $a$ and $b_1-b_2$ commute with $p=\diag(1,0,\dots,0)$.    

By choosing as $p$ the matrices of the form $\diag(0,0,\dots, 1, \dots,0)$ which belong to $\cal O$, we obtain that $a$ and $b_1-b_2$ must commute with any such matrix, which means that they are diagonal. According to Remark 1, the matrices $T a T^{-1}$ and $T (b_1-b_2) T^{-1}$ must be diagonal for any nondegenerate $T$. This means that $a$ and $b_1-b_2$ are scalar matrices. Obviously, the condition $b_1=b_2+2\beta I$ is reduced to $b_1=b_2$ by the transform \eqref{bbb}. As the result, we arrive at equation \ref{P20}.

\subsection{Case $\kappa = 1$}

Like in the case $\kappa=0$, it is sufficient to analyze equations (\ref{eLc}) for $j=0,1,2$ and $3$. 

Let $p=\diag(1,0,\dots,0)$. Then we find for $j=0$ that $c_{0,11}=c_{0,12}=c_{0,21}=0$ and the block $c_{0,22}$ remains arbitrary. 

For $j=1$ we have
\[
 c_{1,11}=-\frac{1}{6}(b_{1,11}+b_{2,11}+z_0),\qquad c_{1,12}=-\frac{1}{4}b_{2,12},
\]
both blocks $c_{1,21}$ and $c_{1,22}$ are arbitrary and $b_1$ should satisfy the restriction
\[
 b_{1,21}=0.
\]
Now we repeat the reasoning from section \ref{kap0}. By choosing the matrices $\diag(0,\dots,1,\dots,0)$ as $p$, we obtain that $b_1$ must be diagonal. Then we prove that it is scalar. Finally, we make use of the transformation \eqref{bbb} and set $b_1=0$.

Next, for $j=2$ all blocks $c_2$ are uniquely determined:
\begin{gather*}
 c_{2,11}=-\frac{1}{4}(1+a_{11}),\qquad
 c_{2,12}=\frac{1}{12}(b_{2,12}c_{0,22}-4a_{12}),\\
 c_{2,21}=\frac{1}{3}(a_{21}+c_{0,22}b_{2,21}+3c_{0,22}c_{1,21}),\\
 c_{2,22}=\frac{1}{2}(a_{22}+z_0c_{0,22}+c_{0,22}b_{2,22}
  +[c_{0,22},c_{1,22}])+c^3_{0,22}.
\end{gather*}

For $j=3$, the blocks $c_{3,11}$ and $c_{3,12}$ remain arbitrary, and the blocks $c_{3,21}$ are $c_{3,22}$ uniquely determined. Moreover, the following relations appear as the solvability conditions for the linear system:
\[
 b_{2,12}(b_{2,21}+4c_{1,21})=0,\qquad b_{2,12}(z_0+b_{2,22}+2c_{1,22}+2c^2_{0,22})=0.
\]
Taking the arbitrariness of $c_{1,22}$ and $c_{0,22}$ into account, this amounts just to the condition $b_{2,12}=0$ which implies, as before, that the matrix $b_2$ is also scalar and then it also can be set to zero. There are no resonance values for $j>3$ and all other coefficients of the series \eqref{y} are uniquely found.  Thus, equations (\ref{eLc}) with $\kappa=1$ are solvable for vanishing $b_1$ and $b_2$, and with the arbitrariness in the blocks $c_{0,22}$, $c_{1,21}$, $c_{1,22}$, $c_{3,11}$ and $c_{3,12}$. Like in the previous case, this gives the required number of arbitrary constants. There are no restrictions for the matrix $a$, since it is not involved in the solvability conditions at all, and we obtain equation \ref{P21}. 

\subsection{Case $\kappa = 2$}

In this case we have to analyze equations for $j=0,1,2,3$ and $4$. 

For $j=0$ we find that $c_{0,11}=c_{0,12}=0$. In contrast to the previous cases, two blocks $c_{0,21}$ and $c_{0,22}$ remain arbitrary. 

For $j=1$ we have
\[
 c_{1,11}=-\frac{1}{6}(b_{1,11}+b_{2,11}+z_0),\qquad 
 c_{1,12}=-\frac{1}{6}b_{2,12},\qquad
 c_{1,21}=\frac{1}{2}b_{1,21}+c_{0,22}c_{0,21},
\]
the block $c_{1,22}$ is arbitrary, no solvability condition appears.

For $j=2$ all blocks of the matrix $c_2$ are found uniquely:
\begin{gather*}
 c_{2,11}=-\frac{1}{4}(a_{11}+1)+\frac{1}{12}(b_{2,12}-3b_{1,12})c_{0,21},\qquad
 c_{2,12}=-\frac{1}{6}(a_{12}+b_{1,12}c_{0,22}),\\
 c_{2,21}=\frac{1}{6}\big(a_{21}+b_{1,22}c_{0,21}-c_{0,21}b_{1,11}+c_{0,22}(2b_{1,21}+b_{2,21})\big)
    +c^2_{0,22}c_{0,21},\\
 c_{2,22}=\frac{1}{2}(a_{22}+z_0c_{0,22}+b_{1,22}c_{0,22}+c_{0,22}b_{2,22})
  +[c_{0,22},c_{1,22}]+\frac{1}{6}c_{0,21}b_{2,12}+c^3_{0,22}.    
\end{gather*}

For $j=3$ the block $c_{3,11}$ is arbitrary, the other blocks are uniquely determined and the first solvability condition appears:
\[
 3b_{1,12}b_{1,21}-2b_{2,12}b_{1,21}-b_{2,12}b_{2,21}
  +6(b_{1,12}-b_{2,12})c_{0,22}c_{0,21}=0.
\]
Due to arbitrariness of $c_{0,22}$ and $c_{0,21}$ we conclude from here that $b_{1,12}-b_{2,12}=0$. By repeating the same reasoning as in the previous cases we prove that the matrices $b_1$ and $b_2$ coincide up to a scalar matrix; then we can set $b_1=b_2=b$ by transformation \eqref{bbb}. It is easy to see that after this the remaining terms in the above relation cancel out.

Finally, for $j=4$ the block $c_{4,12}$ can be choosen arbitrarily if the following condition is fulfilled:
\[
 b_{11}a_{12}+b_{12}a_{22}=a_{11}b_{12}+a_{12}b_{22}+2b_{12}.
\]
This is nothing but $([b,a]-2b)_{12}=0$, and again, by the reasoning based on the arbitrariness of $p$, we conclude that the matrices $a$ and $b$ must satisfy the relation\footnote{Under the transposition, $\kappa=2$ is replaced with $\kappa=-2$ and the sign of the commutator is changed as well.}
\[
 [b,a]=2b.
\]
The rest blocks of $c_4$ and all subsequent $c_j$ are found uniquely. Thus, under the above condition the coefficients of the series are determined with the arbitrariness in the blocks $c_{0,21}$, $c_{0,22}$, $c_{1,22}$, $c_{3,11}$ and $c_{4,12}$ which have the required total dimension. As the result, we obtain equation \ref{P22}. 

The sufficiency of the obtained conditions on the coefficients $\kappa$, $a$ and $b$ for the existence of the series \eqref{y} with arbitrary matrix $p\in{\cal O}$ and the required number of arbitrary constants is proved by the same reasoning in each of three cases. Since the series \eqref{y} with $p=\diag(1,0,\dots,0)$ exists for any  $\kappa$, $a$ and $b$ satisfying the conditions of the theorem, we can replace them with $\kappa$, $\bar a = T a T^{-1}$ and $\bar b=T b T^{-1}$ where $T$ is any non-degenerate matrix. Then the series $T^{-1} y T$ satisfies the equation with parameters $\kappa$, $a$ and $b$ and it has the arbitrary element of the orbit ${\cal O}$ as $p$. \qquad The theorem is proved.

\section{Reductions of partial differential equations}

\subsection{Matrix nonlinear Scr\"odinger equation}

The matrix NLS equation \eqref{NLS} admits the following reduction (which defines solutions that are invariant with respect to some one-parametric subgroup in the group generated by the Galilean transformation, the translation of $t$ and the conjugation by the matrix exponent of a constant matrix):
\begin{equation}\label{NLS.Gal}
 u=e^rp(z),\qquad v=q(z)e^{-r},\qquad r=\frac{1}{6}(t^3-3xt)-tb,\qquad z=x-\frac{1}{2}t^2,\qquad b\in\Mat_n.     
\end{equation}
One can easily check that this substitution reduces equations (\ref{NLS}) to the following ODE system with respect to $p$ and $q$ (cf. with equation (3.1) in \cite{Retakh_Rubtsov_2010}):
\begin{equation}\label{NLS.red}
 p''= -\frac{1}{2}zp-bp-2pqp,\qquad q''= -\frac{1}{2}zq-qb-2qpq,\qquad b\in\Mat_n. 
\end{equation}
The order of this system can be reduced by one due to the first integral
\begin{equation}\label{NLS.int}
 qp'-q'p=c,\qquad c\in\Mat_n.
\end{equation}
Another way to reduce the order is by passing to the logarithmic derivatives of $p$ and $q$. It turns out that to use both methods we can retain only one of the matrix constants $b$ or $c$. The second one should be chosen scalar. Let us consider these two possibilities.
\medskip

1) Let $b\in\Mat_n$ and $2c=\gamma\in\mathbb C$. Then we introduce the variables
\begin{equation}\label{NLS.sub1}
 f=p'p^{-1},\qquad g=q^{-1}q',\qquad h=2pq,
\end{equation}
and the system (\ref{NLS.red}) takes the form
\[
 f'=-f^2-\frac{1}{2}z-b-h,\qquad g'=-g^2-\frac{1}{2}z-b-h,\qquad h'=fh+hg.
\]
Since $\gamma$ is scalar, the first integral (\ref{NLS.int}) can be expressed in terms of $f,g$ and $h$, too:
\[
 f-g=\gamma h^{-1}.
\]  
By its use, we eliminate $g$ and arrive at equations
\begin{equation}\label{fh0}
 f'=-f^2-\frac{1}{2}z-b-h,\qquad h'=fh+hf-\gamma.
\end{equation}
Finally, we find $h$ from the first equation and substitute into the second one; this brings to equation \ref{P20} for the variable $y=f$ and with $a=\gamma-\frac{1}{2}$. 

\begin{remark} 
The symmetry of equations with respect to $f$ and $g$ immediately implies that $g$ also satisfies the same equation with the free term $-\gamma-\frac{1}{2}$. By combining this with the change of sign $g\to -g$ we obtain the B\"acklund transformation for \ref{P20}. One can prove that iterations of this transformation in terms of $p$ or $q$ are governed by the non-Abelian Toda lattice. Vice versa, it is possible to derive \ref{P20}, starting from the Toda lattice and imposing the corresponding reduction (see \cite{Retakh_Rubtsov_2010}). 
\end{remark}

2) Now let $c\in\Mat_n$ and $b$ be a scalar. In this case, according to \eqref{bbb}, we can set $b=0$. Now we apply another substitution
\begin{equation}\label{NLS.sub2}
 f=p^{-1}p',\qquad g=q'q^{-1},\qquad h=2qp
\end{equation}
(alternatively, we can use (\ref{NLS.sub1}) as before, but replace (\ref{NLS.int}) with another first integral $p'q-pq'=\tilde c$ which appears for $b=0$). As a result, the system (\ref{NLS.red}) with $b=0$ turns into
\[
 f'=-f^2-\frac{1}{2}z-h,\qquad g'=-g^2-\frac{1}{2}z-h,\qquad h'=hf+gh,
\]
and the first integral (\ref{NLS.int}) takes the form
\[
 hf-gh=2c.
\]  
By eliminating $g$, we obtain the system
\begin{equation}\label{fh1}
 f'=-f^2-\frac{1}{2}z-h,\qquad h'=2hf-2c.
\end{equation}
In this case, $f$ satisfies equation \ref{P21} with $\kappa=-1$ and $a=2c-\frac{1}{2}$.

An isomonodromic Lax pair for the system (\ref{NLS.red}) can be easily obtained by extending the substitutions (\ref{NLS.Gal}) with the change $\zeta=\lambda-\tfrac{1}{4}t$ and applying this to the well-known representation $U_t=V_x+[V,U]$ of the system (\ref{NLS}), with the matrices
\[
 U=\begin{pmatrix}  
 \lambda & -v \\
  u & -\lambda
 \end{pmatrix},\quad
 V=-2\lambda U+\begin{pmatrix}
  -vu & v_x \\
  u_x & uv
 \end{pmatrix}.
\]
After some simple algebraic manipulations, this leads to the representation \eqref{AB} for equation (\ref{NLS.red}) with the matrices
\[
 B=\begin{pmatrix}
  \zeta & -q\\ p & -\zeta
 \end{pmatrix},\qquad 
 A=\begin{pmatrix}
  8\zeta^2+4qp+z & -8\zeta q-4q'\\[1mm]
  8\zeta p-4p' & -8\zeta^2-4pq-z-4b
 \end{pmatrix}.
\]
From here, it is possible to obtain also the Lax pairs for the above systems in variables $f,h$, by applying gauge transformations. For the system (\ref{fh0}) we have
\[
B=\begin{pmatrix}
   \zeta+f & -\tfrac{1}{2}h\\ 1 & -\zeta
 \end{pmatrix},\qquad
  A=\begin{pmatrix}
   8\zeta^2+2h+z & -4\zeta h-2hf+2\gamma\\[1mm]
   8\zeta -4f & -8\zeta^2-2h-z-4b
  \end{pmatrix}.
 \]
Note that equation \ref{P20} admits also another representation (\ref{AB}) with 
\[
 B=\begin{pmatrix} 
  \zeta & y\\
  y & -\zeta
 \end{pmatrix},
 \qquad
 A=-4\zeta B+\begin{pmatrix}  
  2y^2+z+2b & -2y'-\frac{a}{\zeta}\\[1mm] 
  2y'-\frac{a}{\zeta} & -2y^2-z-2b
 \end{pmatrix},
\] 
which is equivalent to representation from \cite{Irfan_2012}. We were not able to establish a gauge equivalence between these two Lax pairs.

In the case of the system (\ref{fh1}) we obtain \eqref{AB} with
\[
  B=\begin{pmatrix}
   \zeta & -\tfrac{1}{2}h\\ 1 & -\zeta-f
  \end{pmatrix},\qquad
  A=\begin{pmatrix}
   8\zeta^2+2h+z & -4\zeta h-2hf+4c\\[1mm]
   8\zeta -4f & -8\zeta^2-2h-z
  \end{pmatrix}.
\]
This representation for equation \ref{P21} is equivalent to the Lax pair obtained in the next section starting from the mKdV equation.

\subsection{Matrix mKdV equations}

The scalar P$_2$ equation can be obtained also from the mKdV equation
\[
 u_t=u_{xxx}-6u^2u_x
\]
by the self-similar reduction related with the scaling group. This brings to equation $y'''=6y^2y'+y+zy'$ which reduces to P$_2$ by integration. As usual, the zero curvature representation for mKdV turns into an isomonodromic Lax pair. We describe this procedure in the matrix setting.

The first matrix analog of the mKdV \cite{Marchenko_1986} reads
\begin{equation}
\label{mKdV1}
 u_t= u_{xxx}-3u^2u_x-3u_xu^2
 \end{equation}
and admits the representation $U_t=V_x+[V,U]$ with
\[
 U=\begin{pmatrix}
  \lambda & u\\
  u & -\lambda
 \end{pmatrix},\qquad
 V=4\lambda^2U +
 \begin{pmatrix}
  -2\lambda u^2+[u,u_x] & 2\lambda u_x+u_{xx}-2u^3\\[1mm]
  -2\lambda u_x+u_{xx}-2u^3 & 2\lambda u^2+[u,u_x]
 \end{pmatrix}.
\]
In the matrix case, we combine the scaling group with the group of conjugations and apply the self-similar Ansatz
\begin{equation}\label{mKdV.self}
 u=\varepsilon\tau e^{\log(\tau)d}y(z)e^{-\log(\tau)d},\qquad 
 \tau=t^{-1/3},\qquad z=\varepsilon\tau x,\qquad 3\varepsilon^3=-1,\qquad d\in\Mat_n,
\end{equation}
which turns equation (\ref{mKdV1}) into
\begin{equation}\label{mKdV1.red}
 y'''=3y^2y'+3y'y^2+y+zy'+[d,y].
\end{equation}
In contrast to the scalar case, this equation does not have the first integral even for $d=0$. Nevertheless, the reduction of order is still possible due to a {\em partial} first integral. It turns out that (\ref{mKdV1.red}) admits special solutions described by a second-order equation. This can be easily proved by direct elimination of the third and second derivatives in virtue of equation (\ref{gP2}) of the general form. The coefficients are determined by equating the remaining terms and we obtain that (\ref{mKdV1.red}) is satisfied by the equation \ref{P21} with $a=-\kappa d$. To obtain the corresponding Lax pair, we supplement the substitution (\ref{mKdV.self}) with the change $\lambda=\varepsilon\tau\zeta$. Then the dependence of $U$ and $V$ on $\tau$ is separated out:
\[
 U=\varepsilon\tau e^{\log(\tau)d}Be^{-\log(\tau)d},\qquad V=(\varepsilon\tau)^3e^{\log(\tau)d}Ke^{-\log(\tau)d},
\]
where $B$ and $K$ depend on $\zeta$, $z$, $y$ and $y'$ (the second derivative in $K$ is replaced according to \ref{P21}).  For the derivations, this change yields
\[
 \partial_t=(\varepsilon\tau)^3(\tau\partial_\tau+z\partial_z-\zeta\partial_\zeta),\qquad
 \partial_x=\varepsilon\tau\partial_z.
\]
As a result, the equation for $U$ and $V$ takes the form
\[
 -\zeta B_\zeta+zB'+B+[dI,B]=K'+[K,B]
\]
and further change $A=-\zeta^{-1}(K-zB-dI)$ brings to the standard form of the Lax pair \eqref{AB}. The calculations by this scheme give, for equation \ref{P21},
\[
 B=\begin{pmatrix}
  \zeta & y\\
  y & -\zeta
 \end{pmatrix},\qquad
 A= -4\zeta B +\begin{pmatrix}
  2y^2+z & -2y'\\
  2y' & -2y^2-z
 \end{pmatrix}
 -\zeta^{-1}(\kappa[y,y']+a)
 \begin{pmatrix} \kappa & 1\\ 1 &\kappa \end{pmatrix}.
\]

The second matrix analog of mKdV was introduced in \cite{Khalilov_Khruslov_1990}. It turns out that, in contrast to (\ref{mKdV1}), an arbitrary matrix constant can be introduced directly into this equation:
\begin{equation}\label{mKdV2}
 u_t= u_{xxx}+3[u,u_{xx}]-6uu_xu-3(u_x+u^2)c-3c(u_x-u^2),\qquad c\in\Mat_n.
\end{equation}
The origin of this constant is related with the Miura map for the matrix KdV equation which is constructed by solution of the linear Schr\"odinger equation
$$
\psi''+ v \psi + \psi c = 0, \qquad c\in\Mat_n.
$$
The zero curvature representation for (\ref{mKdV2}) is given by the matrices
\begin{gather*}
U=\begin{pmatrix}
  0 & 1\\
  c-\lambda & -2u
 \end{pmatrix},\\[1.5mm]
V=2\begin{pmatrix}
  2u(c-\lambda) & 
  -u_x-u^2-c-2\lambda\\[1mm]
  (u_x-u^2-c-2\lambda)(c-\lambda) & 
  2\lambda u-u_{xx}-[u,u_x]+2u^3+2cu+2uc
 \end{pmatrix}.
\end{gather*}
At first glance, introducing the terms with $c$ into (\ref{mKdV2}) makes the self-similar substitution (\ref{mKdV.self}) impossible, because the homogeneity of the equation is violated. However, we can complement this substitution as follows:
\[
 c=(\varepsilon\tau)^2e^{\log(\tau)d}c_0e^{-\log(\tau)d},
\]
where $c_0$ is an arbitrary constant matrix such that $2c_0+[d,c_0]=0$. Indeed, the differentiation with respect to $\tau$ shows that this relation implies that the matrix $c$ is constant as well (notice, that for scalars this gives $c=c_0=0$, that is, this trick is only possible in the non-abelian case). As the result, (\ref{mKdV2}) is reduced to the equation
\begin{equation}\label{mKdV2.red}
 y'''=3[y'',y]+6yy'y+y+zy'+3(y'+y^2)c_0+3c_0(y'-y^2)+[d,y]. 
\end{equation}
Like in the case of equation (\ref{mKdV1.red}), its order can be reduced by a partial first integral. By eliminating $y'''$ and $y''$ in virtue of equation of the form (\ref{gP2}) we find that if $\kappa=-2$, $3c_0=b$ and $d=-a$ then (\ref{mKdV2.red}) is a consequence of equation \ref{P22}. Further manipulations with the matrices $U$ and $V$ differ from the previous example only by insignificant details (one has to set $\lambda=(\varepsilon\tau)^2\zeta$ and to apply the conjugation by a suitable constant matrix), and this leads to the Lax representation \eqref{AB} with
\begin{gather*}
 B=\begin{pmatrix}
  0 & 1\\
  \frac{1}{3}b-\zeta & -2y
 \end{pmatrix},\\[1.5mm]
 A= \frac{1}{\zeta}\begin{pmatrix}
  2\zeta y-\frac{2}{3}yb-\frac{1}{2}(a+1) & 2\zeta+y'+y^2+\frac{1}{3}b+\frac{z}{2}\\[1mm]
  (2\zeta-y'+y^2+\frac{1}{3}b+\frac{z}{2})(\frac{1}{3}b-\zeta) & 
  -2\zeta y -[y,y']+\frac{1}{3}by+\frac{1}{3}yb+\frac{1}{2}a
 \end{pmatrix},
\end{gather*}
for equation \ref{P22} with $\kappa=-2$.

\section{Conclusion}

We have demonstrated that there are at least three matrix generalizations of the second Painlev\'e equation that satisfy the Painlev\'e--Kovalevskaya test and admit isomonodromic Lax pairs. Of course, a similar diversity should be expected for other Painlev\'e equations, too. Although the literature on their non-Abelian generalizations is quite rich, the question on the number of different versions remains open, in particular, the question of how the non-Abelian constants can enter into the equations.

As one of examples, we present the following matrix version of the P$_4$ equation
\[
 y''= \frac{y'^2}{2\, y}  +\frac{3}{2}y^3+4zy^2 +2(z^2-\alpha)y+\frac{\beta}{y}, 
\]
which contains the matrix constant $c\in\Mat_n$ and the scalar constant $\alpha\in{\mathbb C}$:
\[
 y''= \frac{1}{2}(y'+2c)y^{-1}(y'-2c) +\frac{1}{2}[y,y'] +\frac{3}{2}y^3+4zy^2 +2(z^2-\alpha)y+cy+yc.
\]
This equation can be obtained by eliminating $h$ from the equivalent system
\[
 -y'=y^2+2hy+2zy+2c,\qquad h'=h^2+yh+hy+2zh+\alpha-1
\]
which admits the representation \eqref{AB} with the matrices
\[
 B=\begin{pmatrix}
  -2\zeta^2 & 2\zeta \\
  \zeta h & y-h
 \end{pmatrix},\qquad
 A=\begin{pmatrix}
  4\zeta^3-2\zeta(h+2z)+2\zeta^{-1}c & -4\zeta^2+2y+2h+4z \\
  -2\zeta^2h-hy-\alpha+1 & 2\zeta h+\alpha\zeta^{-1}
 \end{pmatrix}.
\]
This example, as well as a set of other versions of P$_4$, can be obtained by self-similar reduction from equations of nonlinear Schr\"odinger type (in this case we have used equation S$'_5(0,1)$ from \cite{Adler_Sokolov_2020}). For such equations, the reduction procedure brings to systems of two second order equations and  further reduction of order is usually more complicated than in the scalar setting.

%-------------------------------------------------------------------------------
\subsubsection*{Acknowledgements}

The authors are grateful to R.Conte and V.Rubtsov for useful discussions. This work was carried out under the State Assignment 0029-2021-0004 (Quantum field theory) of the Ministry of Science and Higher Education of the Russian Federation.
 
%-------------------------------------------------------------------------------


\begin{thebibliography}{99}
\bibitem{Balandin_Sokolov_1998} S.P. Balandin, V.V. Sokolov. 
 On the Painlev\'e test for non-Abelian equations. {\em Phys. Lett. A \bf 246:3--4} (1998) 
 \href{http://dx.doi.org/10.1016/S0375-9601(98)00336-3}{267--272}.

\bibitem{Retakh_Rubtsov_2010} V. Retakh, V. Rubtsov. 
 Noncommutative Toda chain, Hankel quasideterminants and Painlev\'e II equation. {\em J. Phys. A \bfseries 43} (2010) 
 \href{https://doi.org/10.1088/1751-8113/43/50/505204}{505204}. 

\bibitem{Irfan_2012} M. Irfan. 
 Lax pair representation and Darboux transformation of noncommutative Painlev\'e's second equation. 
 {\em Journal of Geometry and Physics \bf 62} (2012) 
 \href{https://doi.org/10.1016/j.geomphys.2012.01.008}{1575--1582}.
 
\bibitem{Marchenko_1986} V.A. Marchenko. Nonlinear equations and operator algebras. Kiev, Naukova dumka, 1986. 

\bibitem{OS_1998a} P.J. Olver, V.V. Sokolov.
 Integrable evolution equations on associative algebras.
 {\em Commun. Math. Phys. \bf 193:2} (1998) \href{http://dx.doi.org/10.1007/s002200050328}{245--268}.

\bibitem{Khalilov_Khruslov_1990} F.A. Khalilov, E.Ya. Khruslov.
 Matrix generalization of the modified Korteweg--de Vries equation.
 {\em Inverse Problems \bf 6:2} (1990) \href{http://dx.doi.org/10.1088/0266-5611/6/2/004}{193--204}.
 
 \bibitem{Adler_Sokolov_2020} V.E. Adler, V.V. Sokolov. 
 Non-Abelian evolution systems with conservation laws. 
 \href{https://arxiv.org/abs/2008.09174}{arXiv:2008.09174}. 
\end{thebibliography}
\end{document}